\documentclass[twocolumn,accepted=2019-07-10]{quantumarticle}

\usepackage{amsmath}
\usepackage{amsthm}
\usepackage{amssymb}
\usepackage{amsfonts}
\usepackage{mathrsfs}
\usepackage{pgfplots}
\usepackage{xparse} 
\usepackage{graphicx}
\usepackage{color}
\usepackage{hyperref}
\usepackage{qcircuit}
\usepackage[normalem]{ulem}

\hypersetup{
	colorlinks,
	linkcolor={blue},
	citecolor={blue},
	urlcolor={blue}
}

\usepackage{cleveref}

\newcommand{\tV}{\vert\kern-0.25ex\vert\kern-0.25ex\vert}

\newcommand{\tn}[1]{^{\otimes #1}}

\newcommand{\bb}{\mathbb}

\newcommand{\norm}[1]{\ensuremath{\lVert #1 \rVert}}
\newcommand{\ket}[1]{\ensuremath{\lvert #1 \rangle}}

\theoremstyle{plain}
\newtheorem{thm}{Theorem}

\newtheorem{prop}[thm]{Proposition}

\theoremstyle{definition}

\theoremstyle{remark}

\definecolor{nblue}{rgb}{0.2,0.2,0.7}
\definecolor{ngreen}{rgb}{0.1,0.5,0.1}
\definecolor{nred}{rgb}{0.8,0.2,0.2}
\definecolor{nblack}{rgb}{0,0,0}

\newcommand{\hide}[1]{}

\usepackage{tcolorbox}
\usepackage{dsfont}
\usepackage{enumerate}
\usepackage[ruled,vlined]{algorithm2e}
\usepackage{ulem}
\usepackage{cite}
\usepackage{subfig}
\usepackage{floatrow}
\usepackage{multicol}
\usepackage[numbers,sort&compress]{natbib}

\crefname{algocf}{algorithm}{algorithms}
\Crefname{algocf}{Algorithm}{Algorithms}

\crefname{algocf}{algorithm}{algorithms}
\Crefname{algocf}{Algorithm}{Algorithms}

\newcommand{\hd}[1]{\vspace{2mm} \noindent \textbf{#1}}

\newcommand{\inner}[2]{\ensuremath{\langle #1| #2 \rangle}}

\newcommand{\id}{\mathds{1}}

\begin{document}
\title{Clifford recompilation for faster classical simulation of quantum circuits}
\author[1,3]{Hammam Qassim} \email{hqassim@uwaterloo.ca} 
\author[2,3]{Joel J. Wallman}
\author[1,2,3]{Joseph Emerson}
\affiliation[1]{Department of Physics and Astronomy, University of Waterloo, Waterloo, Canada}
\affiliation[2]{Department of Applied Mathematics, University of Waterloo, Waterloo, Canada}
\affiliation[3]{Institute for Quantum Computing, Waterloo, Canada}
\date{\today}
\twocolumn[
  \begin{@twocolumnfalse}
    \maketitle
\begin{abstract}
Simulating quantum circuits classically is an important area of research in quantum information, with applications in computational complexity and validation of quantum devices.
One of the state-of-the-art simulators, that of Bravyi et al, utilizes a randomized sparsification technique to approximate the output state of a quantum circuit by a stabilizer sum with a reduced number of terms.
In this paper, we describe an improved Monte Carlo algorithm for performing randomized sparsification. 
This algorithm reduces the runtime of computing the approximate state by the factor $\ell/m$, where $\ell$ and $m$ are respectively the total and non-Clifford gate counts.
The main technique is a circuit recompilation routine based on manipulating exponentiated Pauli operators.
The recompilation routine also facilitates numerical search for Clifford decompositions of products of non-Clifford gates, which can further reduce the runtime in certain cases by reducing the 1-norm of the vector of expansion, $\norm{a}_1$.
It may additionally lead to a framework for optimizing circuit implementations over a gate-set, reducing the overhead for state-injection in fault-tolerant implementations. 
We provide a concise exposition of randomized sparsification, and describe how to use it to estimate circuit amplitudes in a way which can be generalized to a broader class of gates and states.
This latter method can be used to obtain additive error estimates of circuit probabilities with a faster runtime than the full techniques of Bravyi et al.
Such estimates are useful for validating near-term quantum devices provided that the target probability is not exponentially small.  
\end{abstract}
  \end{@twocolumnfalse}
]

\section{Introduction} 
Digital simulation of quantum phenomena is a central problem in computational physics, with deep theoretical and practical significance.
On the theory side, research into this problem has produced surprising connections between different branches of mathematics and physics, while providing major insight into the inner-workings of quantum theory.
On the practical side, there is high demand for faster and more accurate simulations in fields such as quantum computing, materials science, and quantum chemistry.
In these disciplines, it is desirable to have a toolkit of simulation techniques which can handle the wide variety of quantum systems encountered.

Quantum computers are widely believed to be hard to simulate classically.
Nevertheless, this hardness is only asymptotic, and it is always possible in principle to simulate larger and larger quantum computers by optimizing the usage of classical resources.
Efforts to demonstrate such optimizations have surged recently \cite{de2018massively,smelyanskiy2016qhipster,chen201864,boixo2017simulation,chen2018classical}.
These efforts are partially motivated by upcoming quantum computers containing on the order of 100 qubits \cite{Preskill2018,Neill2018}.
Since these prototypes are likely to be considerably noisy, the ability to simulate their ideal counterparts is essential for validating their performance and to aid in their design and debugging.
Furthermore, simulation of quantum circuits is directly tied to establishing when quantum computers are useful, as demonstrating quantum supremacy requires sampling from distributions which are provably hard to (approximately) sample from using classical means \cite{Harrow_2017}.

Near-Clifford circuits constitute an important class of quantum circuits.
Fault-tolerant implementations of quantum computing in the near-term are likely to contain a small number of non-Clifford gates. 
This is the case since the overhead associated with implementing non-Clifford gates fault-tolerantly is huge, for example accounting for nearly $90\%$ of the number of physical qubits in surface code implementations of Shor's algorithm \cite{Fowler2012}.  
A host of classical simulation schemes have recently been proposed for simulating near-Clifford circuits, with runtimes exhibiting mild exponential dependence on $n$ \cite{garcia2014geometry,bravyi2016trading,bravyigosset,bravyi2018simulation,howard2017application,bennink2017monte, heinrich2018robustness}.
State-of-the-art techniques rely on decomposing the $n$-qubit output state of a near-Clifford circuit $U$ as a superposition of stabilizer states \cite{bravyigosset,bravyi2018simulation},
\begin{align}\label{eq: decomposition}
\ket{\psi} \equiv U\ket{0^n}=\sum_{i=1}^N a_i \ket{v_i},
\end{align}
with $N \ll 2^n$.
These techniques allow for approximate sampling from the outcome distribution of the circuit, with the runtime for producing a single sample scaling like $O(N \epsilon^{-2})$, where $\epsilon$ is the error in total variation distance \cite{bravyigosset,bravyi2018simulation}.
The minimum number of terms over all such decompositions of $\ket{\psi}$ is called its \textit{stabilizer rank}, and is a measure of the complexity of simulating $\ket{\psi}$.
Finding stabilizer decompositions with a small number of terms is typically quite hard, even for simple product states \cite{bravyigosset,bravyi2018simulation}.
This motivates using \textit{approximate} stabilizer decompositions. 
It was shown in \cite{bravyi2018simulation} that an approximation to $\ket{\psi}$ with a smaller number of stabilizer terms can be obtained by ``randomly sparsifying" any decomposition of the type in \cref{eq: decomposition}.
This is done by generating $k$ random stabilizer states, each chosen independently to be $\ket{v_i}$ with probability $|a_i| / \norm{a}_1$,
then setting the approximate state $\ket{\tilde{\psi}}$ to be their equal superposition (with a uniform scaling factor).
If $k$ is large enough (roughly $\norm{a}_1^2 \epsilon^{-2}$), then $\norm{\psi - \tilde{\psi}} \leq \epsilon$ with high probability.
Typically $k \ll N$, so that using the approximate state $\ket{\tilde{\psi}}$ for simulation purposes is preferable.
 
Randomized sparsification can be performed using a sum-over-Cliffords Monte Carlo method \cite{bravyi2018simulation}, which we briefly review before stating our contribution. 
Given a circuit $U$, each non-Clifford gate in $U$ can be expanded as a linear combination of Clifford gates.
The entire circuit can then be expressed as
\begin{align}
U = \sum_i a_i C_i,
\end{align} 
for a collection of (typically exponentially many) Clifford circuits $C_i$.
The sum-over-Cliffords method uses this decomposition to compute the stabilizer states in \cref{eq: decomposition} while preserving the relative phases between them.
For each term, a phase-sensitive Clifford simulator is used to convert the description of the circuit $C_i$ to a description of the stabilizer state $\ket{v_i} = C_i\ket{0^n}$ while keeping track of the global phase.
Similarly, the stabilizer terms of the approximate state $\ket{\tilde{\psi}}$ can be computed by applying circuit $C_i$ with probability $|a_i|/\norm{a}_1$.
The Clifford simulator computes the so-called CH-form of $C_i \ket{0^n}$ by performing an update for each gate in $C_i$. 
This computation takes time $O(\ell n^2)$, where $\ell$ is the total number of gates in the circuit.
Thus to compute the CH-form for all stabilizer terms comprising the approximate state $\ket{\tilde{\psi}}$, the method in \cite{bravyi2018simulation} runs in time $O(k \ell n^2)$.

In the current paper, we report an improved variant of randomized sparsification based on recompiling the circuit using exponentiated Pauli operators. 
This variant removes the dependence on the total number of gates in the circuit.
More precisely, it reduces the time cost of computing the CH description of $\ket{v_i}$ from $O(\ell n^2)$ to $O(mn^2)$, where $m$ is the non-Clifford gate count in $U$.
This reduces the runtime of key subroutines in \cite{bravyi2018simulation} by the factor $\ell/m$, which can be significant for circuits in which the number of Clifford gates is very large compared to non-Clifford gates.
This is relevant, for example, in the case of the QAOA circuit simulated in \cite{bravyi2018simulation}, which contains hundreds of Clifford gates and only $64$ T gates. 
The above savings are achieved at an additional memory cost of $O(nm)$ bits, which is negligible in the regime of interest.

As a side benefit, the Clifford recompilation method makes it easier to numerically search for decompositions of the kind in \cref{eq: decomposition} with a reduced 1-norm. 
As the exponential scaling of the runtimes in \cite{bravyi2018simulation} is determined by $\norm{a}_1^2$, this allows for a significant speed-up in some cases. 
In addition to reducing $\norm{a}_1^2$, it can even be the case that some of the gates in the recompiled circuit cancel out, which
suggests a systematic way of optimizing the compilation step to reduce the number of non-Clifford gates in a given circuit, with the possibility of significant overhead savings in fault-tolerant implementations.
 
Another use of Clifford recompilation is in estimating circuit amplitudes.
It is straight-forward to see that the random variable $\inner{x}{\tilde{\psi}}$ is an unbiased estimator of the amplitude $\inner{x}{\psi}$, where $x \in \bb{F}_2^n$ denotes a computational basis state.
Clifford recompilation reduces the complexity of evaluating $\inner{x}{\tilde{\psi}}$ by the factor $\ell/m$, in comparison to \cite{bravyi2018simulation}.
Estimates of the amplitude obtained in this way are used in the Metropolis simulator in this latter work.
They can also be used to estimate the probability $|\inner{x}{\psi}|^2$ up to additive error $\epsilon$ with a runtime scaling like $O(\norm{a}_1^2 \epsilon^{-2})$.
Learning probabilities in this way can be generalized to a wider range of quantum circuits.
It is also much less memory intensive, and is better suited for massive parallelization and GPU-computing.
However, in general such polynomial precision estimates are only useful as long as the target probability is not exponentially small.

\textit{Notation.} We write $X$, $Y$, and $Z$ for the single qubit Pauli matrices, and typically use the letters $P, Q$ and $W$ to refer to Hermitian elements of the n-qubit Pauli group. 
The shorthand $M_{b:a}$, with $b>a$, denotes the product $M_b M_{b-1} \dots M_a$.

\section{Improved randomized sparsification}\label{sec: improved sparsification}
\hd{The CH-form of stabilizer states.} We begin by reviewing the CH-form of stabilizer states.
A C-type Clifford is one which satisfies $U_C \ket{0^n} = \ket{0^n}$. 
An H-type Clifford consists of Hadamard gates on some subset of the qubits.
Any stabilizer state of $n$-qubits can be written as
\begin{align}
\ket{\phi} = \omega U_C U_H \ket{s},
\end{align}
where $U_C$ is a C-type Clifford, $U_H$ is an H-type Clifford, $s \in \bb{Z}_2^n$, and $\omega$ is a complex number.
$U_C$ can be described by a stabilizer tableau specifying its action on the Pauli operators $X_i$ and $Z_i$ for $i=1, \dots, n$, 
and $U_H$ can be described by an $n$-bit string specifying which qubits are acted on.
As shown in \cite{bravyi2018simulation}, the tuple $(\omega,U_C,U_H)$ describing the state $\ket{\phi}$ can be updated in time $O(n)$ under acting on $\ket{\phi}$ with a gate from the set $\{\text{CNOT, CPHASE, PHASE}\}$, and in time $O(n^2)$ under applying a Hadamard gate. 
For our circuit recompilation, we need the following fact concerning updates of the CH-form under Clifford gates of a particular exponential form.
\begin{prop}\label{exponential Pauli updates}
The tuple $(\omega,U_C,U_H)$ can be updated in time $O(n^2)$ under applying a Clifford gate of the form $\exp(i \theta P)$, where $P$ is a Pauli and $\theta \in (\pi/4)\bb{Z}$.
\end{prop}
\begin{proof}
Without loss of generality we assume that $\theta \in \{\pi/2, \pi/4 \}$.
We can compute a Pauli $P'$ such that $(U_C U_H)^\dagger P U_C U_H = P'$ in time $O(n^2)$, using the stabilizer tableau of $U_C$ and the fact that $HXH=Z$.
Given such $P'$, the state can be updated in time $O(n^2)$.
This follows from the following observations.
If $\theta = \pi/2$, then $\exp(i \theta P) wU_CU_H\ket{s}=wU_CU_H(iP')\ket{s}=w' U_C U_H \ket{s'}$, for some $w'$ and $s'$ which can be computed in time $O(n)$.
If $\theta = \pi/4$, then $\exp(i \theta P) wU_CU_H\ket{s}=wU_CU_H\exp( i \frac{\pi}{4} P')\ket{s}=w' U_C U_H (\ket{s} +i^{\delta} \ket{s'})$, for some $\delta \in \bb{Z}_4$, $\omega'$, and $s'$, all of which can be computed in time $O(n)$.
If $s=s'$ then the update is trivial.
If $s \neq s'$, then the update can be computed in time $O(n^2)$ by combining Proposition 4 in \cite{bravyi2018simulation} with the update rule for $U_C$ under right-multiplication by gates from the set $\{\text{CNOT, CPHASE, PHASE}\}$, which is given in the same section in \cite{bravyi2018simulation}. 
\end{proof}
   
\hd{Clifford recompilation}. We now describe recompiling the circuit in order to reduce the complexity of computing the CH-forms of the output state.
Consider a quantum circuit $U$ consisting of $\ell$ gates, $m$ of which are non-Clifford gates. 
We can write $U$ in the form
\begin{align}\label{generic circuit}
U = C_m U_m C_{m-1} U_{m-1} \dots C_1 U_1 C_0,
\end{align}
where $C_i$ are Clifford circuits, and $U_i$ are non-Clifford gates.
We can then ``pull" all non-Clifford gates to the end of the circuit
\begin{align}\label{eq: circuit}
U = U'_m \dots U'_1 C_m  C_{m-1} \dots C_1  C_0,
\end{align}
where $U_i' =  C_{m:i}U_iC_{m:i}^\dagger$.
We can absorb the action of the Clifford gates into the input state, so that the output state is given by
\begin{align}\label{eq: reduction}
\ket{\psi} = U \ket{0^n}= U' \ket{\phi}, \quad  U' \equiv U_{m:1}',
\end{align}
where $\ket{\phi}\equiv C_{m:0} \ket{0^n}$ is a stabilizer state. 
Note that the CH-form of $\ket{\phi}$ can be computed in time $O(cn^2)$, where $c \equiv \ell - m$ is the number of Clifford gates in $U$.
Storing the CH-form of $\ket{\psi}$ in memory requires $O(n^2)$ bits.

Next we discuss how to compute the $U_i'$.
Suppose we are given a Clifford decomposition of each non-Clifford gate in $U$;
\begin{align}\label{eq: clifford decomposition}
U_i = \sum_{j}b_{ij}V_{ij},
\end{align}
where each $V_{ij}$ is a product of a small number of elementary Clifford gates.
We will use exponentiated Pauli matrices to compute and store $U_i'$ efficiently.
First note that each $V_{ij}$ can be written as a product of Clifford gates of the form $\exp(i \theta P)$ where $P$ is a Pauli and $\theta \in (\pi/4)\bb{Z}$ \footnote{For example, the Hadamard, PHASE, and CZ gates can be written as
\begin{align}\label{eq: Clifford exp form}
H_j &= e^{i \pi/2} e^{-i \frac{\pi}{2}Z_j} e^{i \frac{\pi}{4}Y_j}, \nonumber\\
S_j &= e^{i \pi/4} e^{-i \frac{\pi}{4}Z_j},  \nonumber\\
CZ_{jk}&= e^{i \pi/4}e^{i \frac{\pi}{4}Z_j Z_k}e^{-i \frac{\pi}{4}Z_j}e^{-i \frac{\pi}{4}Z_k}.
\end{align}
Similar identities can be derived for any Clifford gate acting on a fixed number of qubits.
}.
We can therefore write
\begin{align}
U_i = \sum_j a_{ij} \prod_{k=1}^{O(1)} \exp(i\theta_{ijk} P_{ijk}),
\end{align}
where $a_{ij}$ absorbs any phases from the elementary gates, and $\{ P_{ijk}\}$ is some collection of Pauli matrices.
Then
\begin{align}\label{eq: clifford decomposition 2}
U_i'
&= \sum_j b_{ij} C_{m:i}\left[\prod_{k} \exp(i\theta_{ijk} P_{ijk})\right]C_{m:i}^\dagger \notag\\
&= \sum_j b_{ij} \prod_{k} \exp(i\theta_{ijk} C_{m:i}P_{ijk}C_{m:i}^\dagger) \notag\\
&= \sum_j b_{ij} \prod_{k} \exp(i\theta_{ijk}' P_{ijk}') \notag\\
&= \sum_j b_{ij} V'_{ij}.
\end{align}
where 
\begin{align}
V'_{ij} \equiv \prod_{k} \exp(i\theta_{ijk}' P_{ijk}'),
\end{align}
$\theta_{ijk}'\in\{\pm \theta_{ijk}\}$, and the $P_{ijk}'$ are Pauli matrices each of which can be computed efficiently in time $O(c)$.
Storing this description of $U'$ requires $O(nm)$ bits of memory, since each Pauli matrix requires $O(n)$ bits to specify.
We remark that the use of Pauli exponential version of Clifford and non-Clifford gates has recently found application in the context of error-correction \cite{litinski2018game}. 

To conclude, we can write the output state of the circuit $U$ as
\begin{align}
\ket{\psi} 
&= U' \ket{\phi} \notag\\
&= \sum_{i} a_i \ket{u_i},
\end{align}
where each stabilizer state $\ket{u_i}$ is produced by acting on the fixed stabilizer state $\ket{\phi}$ with a circuit consisting of $O(m)$ Clifford gates of the form $\exp(i \theta P)$, $\theta \in (\pi/4)\bb{Z}$.
Combining this with \Cref{exponential Pauli updates} shows that we can produce the CH-form of $\ket{u_i}$ in time $O(m n^2)$, which improves on the $O(\ell n^2)$ required by the method described in \cite{bravyi2018simulation}.
\begin{algorithm}[t]
\LinesNumbered
\SetKwData{Left}{left}\SetKwData{This}{this}\SetKwData{Up}{up}
\SetKwFunction{Union}{Union}\SetKwFunction{FindCompress}{FindCompress}
\SetKwInOut{Input}{Input}\SetKwInOut{Output}{Output}
\Input{The CH-form $(\omega_0, U^0_C, U_H^0 )$ of the stabilizer state $\ket{\phi_y}$. A decomposition of the form in \cref{eq: clifford decomposition 2} for $i=1, \dots, m$.}
\Output{A random vector $\ket{z}$ such that $\bb{E}(\ket{z}) = \ket{\psi}$.}
\Begin{
Set $(\omega, U_C, U_H) \leftarrow (\omega_0, U^0_C, U_H^0 ) $

\For{i =1, \dots, m}{
\small{
Choose $V'_{ij}$ with probability $p_i(j)\equiv|b_{ij}|/\norm{b_i}_1$;

Update  $(\omega, U_C, U_H)$ under $V'_{ij}$;

Set $\omega \leftarrow (b_{ij}/|b_{ij}|)\omega$;

}
}
Output $\ket{z} = \norm{b}_1 w U_C U_H \ket{s}$.
}
\bf{end}
\caption{Monte Carlo algorithm for improved randomized sparsification}
\label{alg: improved randomized sparsification}
\end{algorithm}

This improvement is applicable to both the norm estimation and the heuristic Metropolis algorithms of \cite{bravyi2018simulation}. 
The total time cost of computing the CH-forms of the approximate state is $O(k m n^2)$ in our method, vs. $O(k \ell n^2)$ in \cite{bravyi2018simulation}.
The runtime reduction factor of $\ell/m$ is significant because the randomized sparsification subroutine requires computing an exponential number of CH-forms, roughly $k \approx \norm{a}_1^2 \epsilon^{-2}$, in order to obtain an $\epsilon$-approximation to $\ket{\psi}$.
Recall that the factor $\norm{a}_1^2$ is exponential in the number of non-Clifford gates. For a Clifford+T circuit, for example, it is equal to $2^{0.228 m}$, where $m$ is the T count.

\hd{Monte-Carlo method.}
Randomized sparsification (both the original version and our improved variant) can be implemented as a Monte Carlo algorithm.
The state space is the set of CH-tuples $(\omega, U_C, U_H)$.
In our improved variant, the initial state is the CH-form of $\ket{\phi}$.
For each gate $U'_i$ in the reduced circuit, a random Clifford gate is applied to the current state according to \cref{eq: clifford decomposition 2}. 
Namely, $V'_{ij}$ is applied with probability $|b_{ij}|/\norm{b_i}_1$.
The relative phases between the terms are accounted for by updating $\omega$ appropriately. 
After the last gate, the final state is scaled by the factor $\norm{b}_1 \equiv \prod_i \norm{b_i}_1$.
This is summarized in \Cref{alg: improved randomized sparsification}.
A direct calculation shows that the output $\ket{z}$ of this algorithm satisfies
\begin{align}\label{eq: avg of alg 1}
\bb{E}(\ket{z}) = \ket{\psi}.
\end{align}
Moreover, it is shown in \cite{bravyi2018simulation} that the average of $k$ runs of the algorithm is a random vector $\ket{\tilde{\psi}} = k^{-1} \sum_{i=1}^{k} \ket{z_i}$ which satisfies
\begin{align}
\Pr \left[ \norm{\ket{\tilde{\psi}} - \ket{\psi}} \geq \epsilon  \right] \leq \frac{\norm{a}_1^2/k}{\epsilon^{2}}.
\end{align}
For example, taking $k = (2\norm{a}_1/\epsilon)^2$ ensures that the probability of the error exceeding $\epsilon$ is at most $1/4$.
\begin{algorithm}[t]
\LinesNumbered
\SetKwData{Left}{left}\SetKwData{This}{this}\SetKwData{Up}{up}
\SetKwFunction{Union}{Union}\SetKwFunction{FindCompress}{FindCompress}
\SetKwInOut{Input}{Input}\SetKwInOut{Output}{Output}
\Input{The CH-form $(\omega_0, U^0_C, U_H^0 )$ of the stabilizer state $\ket{\phi_y}$. A decomposition of the form in \cref{eq: clifford decomposition 2} for $i=1, \dots, m$.}
\Output{A random vector $\ket{z}$ such that $\bb{E}(\ket{z}) = \ket{\psi}$.}
\Begin{
Set $(\omega, U_C, U_H) \leftarrow (\omega_0, U^0_C, U_H^0 ) $

\For{i =1, \dots, m}{
\small{
Compute \cref{eq: adaptive};

Choose $l$ with probability $p'_i(l)\equiv|a'_{il}|/\norm{a'_i}_1$;

Set  $(\omega, U_C, U_H)$ to be the CH-form of $\ket{u_{il}}$;

Set $\omega \leftarrow (a'_{il}/|a'_{il}|)\omega$;

}
}
Output $\ket{z} = \norm{b}_1 w U_C U_H \ket{s}$.
}
\bf{end}
\caption{Adaptive Monte Carlo algorithm for randomized sparsification}
\label{alg: adaptive Monte Carlo}
\end{algorithm}

\hd{Adaptive Monte-Carlo.} The convergence of the estimator in the Monte Carlo of \Cref{alg: improved randomized sparsification} can be improved by adapting the probabilistic Clifford update to the particular state at each time step.
Let the stabilizer state at the $i$th time step of the Monte Carlo be $\ket{u}$, and let
\begin{align}\label{eq: adaptive}
U'_i \ket{u} 
&=  \sum_{j} a_{ij} V'_{ij}\ket{u}, \nonumber\\
&=  \sum_{j} a_{ij} \ket{u_{ij}}, \nonumber\\
&= \sum_{l} a'_{il} \ket{u_{il}},
\end{align}
where all collinear terms in the second sum are collected in the third sum.
Note that computing \cref{eq: adaptive} can be done efficiently by updating the CH description of each term under the Clifford gates $V'_{ij}$.
It is then straightforward to show that $\norm{a'_i}_1 \leq \norm{a_i}_1$.

Consider the modified Monte Carlo in \Cref{alg: adaptive Monte Carlo}.
As before, we have $\bb{E}(\ket{z}) = \ket{\psi}$, and one can take the average of $k$ iterations of this algorithm to approximate $\ket{\psi}$.
As $\norm{a'}_1$ depends on the particular trajectory taken, we need to choose $k \approx \Delta^2 \epsilon^{-2}$, where 
\begin{align}
\Delta = \max_{\text{trajectories}} \norm{a'}_1 \leq \norm{a}_1.
\end{align}
Computing $\Delta$ is generally hard since it involves a maximization over all trajectories. 
Unless a smaller upper bound for $\Delta$ can be learned, the safe solution is to settle for the pessimistic sample size which scales like $\norm{a}_1^2$.
Even with this limitation, the above method improves the precision of the approximation obtained for the same sample size.
We can use heuristics to predict whether adaptive sampling is useful.
Let $\alpha$ be the random variable which takes value $\norm{a'}_1$ with probability equal to the sum of the probabilities of all trajectories consistent with $\norm{a'}_1$.
One can estimate the mean and variance of $\alpha$ by sampling, and heuristically argue that most trajectories have $\norm{a'}_1$ within, say, 3 standard deviations of the mean.
Thus if the mean is much smaller than $\norm{a}_1$, and the variance is not too large, it is a good indication that adaptive Monte Carlo improves the rate at which $\ket{\tilde{\psi}}$ converges to $\ket{\psi}$. 

\hd{Clifford decompositions of products of gates.} Finding a Clifford decomposition of $U_i$ with minimal 1-norm $\norm{a_i}_1$ can be done using a convex optimization software  \cite{howard2017application,bravyi2018simulation}.
These brute-force optimizations can only be performed for one- or two-qubit gates, as it is too computationally intensive for more qubits.

For an operator $A$, denote by $\Omega(A)$ the minimum 1-norm over all Clifford decompositions of $A$.
Then $\Omega$ is sub-multiplicative, i.e.
\begin{align}\label{eq: product submultiplicativity}
\Omega(AB) \leq \Omega(A) \Omega(B).
\end{align}
For two operators $A$ and $B$, let us call a Clifford decomposition of the product $AB$ \textit{contractive} if its 1-norm is strictly less than $\Omega(A)\Omega(B)$.
Finding a contractive Clifford decomposition of $AB$ is generally hard unless $A$ and $B$ are at most 2-qubit gates acting on the same pair of qubits, so that we can use brute-force convex optimization.
An exception is when the two operators can be mapped to the same pair of qubits via a Clifford $C$.
In that case, one can find a contractive decomposition of $C AB C^\dagger$ using convex-optimization, and then apply the inverse $C^\dagger$, using the idea of \cref{eq: clifford decomposition 2} to keep the terms in a form enabling fast CH updates.
This method is well-suited to the case of Clifford+$e^{i\theta P}$ circuits, which we discuss in the next section.
  
Note that strict inequality in \cref{eq: product submultiplicativity} does not hold, since two non-Clifford gates can multiply to a Clifford gate. 
Furthermore, in the next section we give examples where equality holds although the product $AB$ is not a Clifford gate. 

Sub-multiplicativity also holds under the tensor-product, 
\begin{align}
\Omega(A \otimes B ) \leq \Omega(A) \Omega(B).
\end{align}
The question of whether $\Omega$ is multiplicative under the tensor product is open.
For tensor products of 1- and 2-qubit gates, our numerical search does not find any counter examples, and thus we conjecture that $\Omega$ is multiplicative for 1- and 2-qubit gates.
The analogue of $\Omega$ for quantum states was shown in \cite{bravyi2018simulation} to be multiplicative when each tensor factor is a state of at most three qubits, which suggests that a similar fact should hold for unitaries.
However, the proof technique in \cite{bravyi2018simulation} is not straight-forward to generalize to the unitary version, since it relies on the so-called stabilizer-alignment property of quantum states, which (to our knowledge) has no clear analogue for unitaries.
We note that this multiplicative behavior for states, which we also conjecture to hold for unitaries, stands in contrast to other monotones such as the robustness of magic \cite{howard2017application}, and its channel version \cite{bennink2017monte}, which are known to be strictly sub-multiplicative in important cases. 

\section{Example: Clifford+$e^{i\theta P}$ circuits}
As an application of the methods presented so far, we consider circuits where each non-Clifford gate has the form $e^{i\theta P}$, where $P$ is a Pauli operator, and $\theta \in [-\pi,\pi]$.
This class of circuits includes important cases such as Clifford$+T$ circuits, and more generally Clifford+$Z$ rotation circuits.
For a circuit $U$ of this type containing $m$ non-Clifford gates, \cref{eq: reduction} reduces to
\begin{align}\label{cliff and exp pauli amp}
U\ket{0^n} = e^{i \theta_m Q_m} \dots e^{i \theta_1 Q_1}\ket{\phi},
\end{align}
for a list of Hermitian Pauli operators $Q_1, \dots, Q_m$, a list of angles $\theta_1, \dots, \theta_m$, and a stabilizer state $\ket{\phi}$.
Note that the unitary $e^{i\theta P}$ is Clifford if and only if $\theta\in (\pi/4)\bb{Z}$.
Thus we may, without loss of generality, assume that $\theta_i \in (0,\frac{\pi}{4})$, as otherwise we can ``factor out'' Clifford gates $e^{\pm i \frac{\pi}{4}Q_i}$ and absorb them into the state $\ket{\phi}$, at the cost of applying a Clifford mapping to the list of Pauli operators $Q_j$ with $j<i$.
The Clifford mapping can be quickly computed using the relation
\begin{align}
e^{i \frac{\pi}{4}P}Qe^{-i \frac{\pi}{4}P}=
\begin{cases}
Q & \text{if }[P,Q]=0\\
iPQ & \text{if } \{P,Q\}=0.
\end{cases}
\end{align}
We refer to \cref{cliff and exp pauli amp} as the \textit{canonical form} of the output state of a Clifford+$e^{i\theta P}$ circuit.

For a single qubit Pauli operator $P$, the optimal decomposition of $e^{i \theta P}$ into Clifford gates is given by \cite{bravyi2018simulation}
\begin{align}\label{eq: optimal pauli exp decoms}
e^{\pm i \theta P} = (\cos \theta - \sin \theta) \id + (\sqrt{2} \sin \theta) e^{\pm i \frac{\pi}{4} P}.
\end{align}
This decomposition achieves the minimum  1-norm,
\begin{align}\label{eq: optimal pauli exp norm}
\Omega(e^{i \theta P}) = \cos \theta + (\sqrt{2}-1) \sin \theta \quad\quad(0 \leq \theta \leq \frac{\pi}{4}).
\end{align}
This holds for $n$-qubit Pauli operators as well, which follows from i) invariance of $\Omega$ under multiplication by Clifford unitaries, ii) invariance of $\Omega$ under tensoring by the identity, i.e. $\Omega(U\otimes \id)=\Omega(U)$.

\begin{figure*}[t]
\centering
\ffigbox[\FBwidth]
{
\subfloat[One qubit]{\includegraphics[height=9cm,width=8cm]{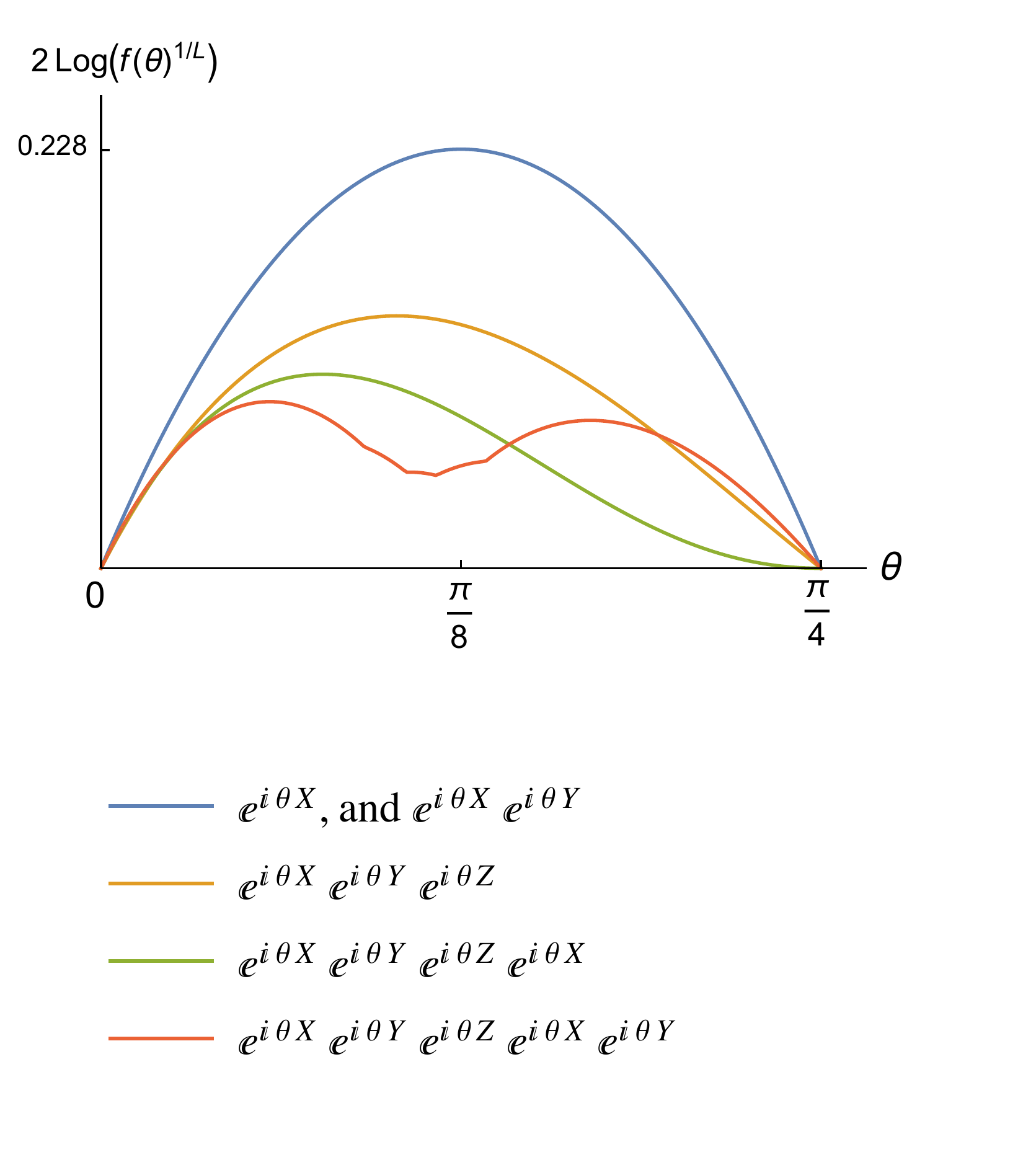}}
\quad
\subfloat[Two qubits]{\includegraphics[height=9cm,width=8cm]{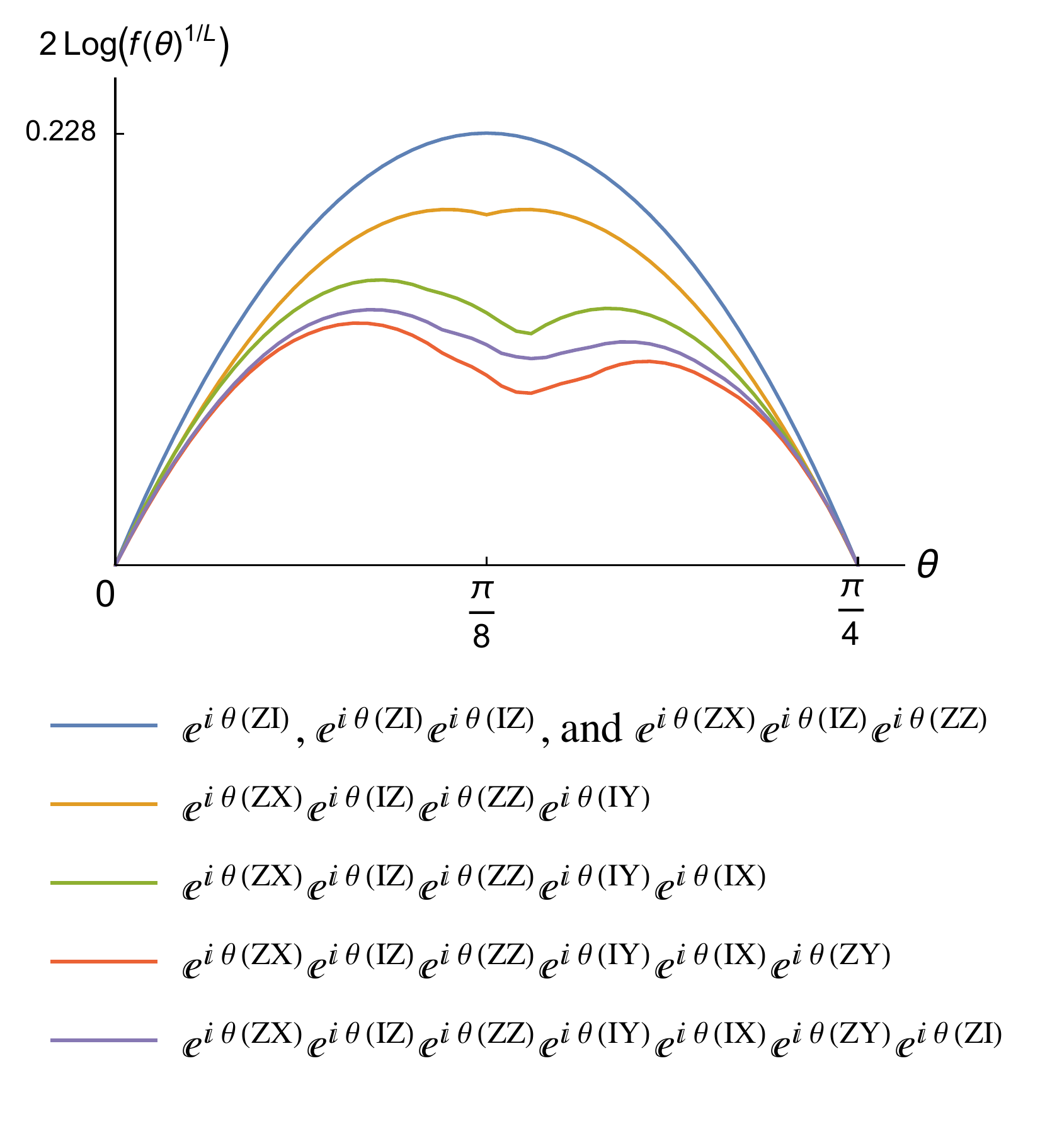}}
}{\caption{Theoretical asymptotic cost of simulating various one- and two-qubit exponential sequences. The angles are chosen to be equal for simplicity.
The function $f(\theta)$ is the 1-norm of the optimal decomposition of the sequence into Clifford gates, and, for a sequence of length $L$, $f(\theta)^{1/L}$ is the effective 1-norm per gate. The function $2 \log(f(\theta)^{1/L})$ is the contribution of each gate to $\alpha$, where $2^{\alpha}$ is the exponential scaling factor in the runtime of the simulation. The data is obtained using a convex optimization software (CVX with the SDPT3 solver).}}
\end{figure*}

\hd{Clifford decompositions of products of exponentiated Pauli operators.}
If a sequence of exponentials in \cref{cliff and exp pauli amp} multiplies to a Clifford gate, it can be commuted past the gates prior to it and absorbed into the input state, effectively reducing the number of non-Clifford gates in the circuit.
As we are dealing with a universal set of gates, it is generally intractable to determine whether such a product belongs to the Clifford group.
In the simplest case of multiplying two gates $e^{i\theta P} e^{i\phi Q}$ with $\theta,\phi \in (0,\pi/4)$, it can be shown that the product is a Clifford gate if and only if $PQ=\pm \id$ and $\theta+\phi \in (\pi/4)\bb{Z}$.
If $PQ\neq \pm \id$, there exists no better decomposition than the trivial one, obtained by decomposing each gate separately;
\begin{align}
\Omega\left(e^{i \theta_1 P} e^{i \theta_2 Q}\right) = \Omega\left(e^{i \theta_1 P}\right) \Omega\left(e^{i \theta_2 Q}\right).
\end{align}
This can be verified numerically for 1 and 2 qubits, and the general case follows from the fact that any pair of $n$-qubit Paulis $P$ and $Q$ can be mapped by a Clifford circuit to a pair of Paulis with support on the first two qubits only, together with the fact that $\Omega(U\otimes \id)=\Omega(U)$.

For products of three or more gates, we can find decompositions with significantly reduced 1-norm in certain cases.
Numerically, we observe that the possibility of 1-norm reduction is correlated with the rank of the binary representation of \hide{dimension of the algebra spanned by }the Paulis.
Specifically, for $x,z \in \bb{Z}_2^{n} $  write $X[x]\equiv X^{x_1}\otimes\dots \otimes X^{x_n}$, $Z[z]\equiv Z^{z_1}\otimes \dots \otimes Z^{z_n}$, and $P_{(x,z)}\equiv i^{x\cdot z} X[x] Z[z]$.
Then numerics in the two qubit case suggest that,
\begin{align}
\Omega\left( \prod_j e^{i \theta_j P_{v_j}} \right) < \prod_j \Omega(  e^{i \theta_j P_{v_j}})
\end{align}
if and only if $\text{rank} (v_1,  v_2, \cdots, v_k) < k$ with rank$(\cdot)$ computed mod 2. 
For example, it is possible to numerically verify that
\begin{align}
\Omega\left(e^{i\theta_1 X} e^{i\theta_2 Y} e^{i\theta_3 Z}\right) < \Omega\left(e^{i \theta_1 X}\right)\Omega\left(e^{i \theta_2 Y}\right)\Omega\left(e^{i \theta_3 Z}\right),
\end{align}
for any choice of angles $\theta_1,\theta_2,\theta_3 \in (0, \pi/4)$.
We can also numerically find the optimal Clifford expansion \footnote{For $0 \leq \theta \leq \pi/8$, the expansion
{\scriptsize
\begin{align}
e^{i\theta X} e^{i\theta Y} e^{i\theta Z} &= (\cos 3\theta - \sin \theta) \id
+ e^{i \pi/4} [\sin 2\theta (\cos \theta - \sin \theta)] SH \nonumber\\
&+ (\sqrt{2} \sin 2\theta \sin \theta) HX + (\sqrt{2} \sin 2\theta \sin \theta) SHS.
\end{align}}
is optimal.
More generally, the Clifford terms in \cref{eq: 3-sequence} vary depending on $(\theta_1, \theta_2, \theta_3)$, and a general expression is too tedious to write down, although finding the optimal decomposition numerically is easy on a case-by-case basis.}, say
\begin{align}\label{eq: 3-sequence}
e^{i\theta_1 X} e^{i\theta_2 Y} e^{i\theta_3 Z} = \sum_{i} \alpha_i C_i.
\end{align}
Therefore for any triple of n-qubit Paulis $(P,Q,W)$ which are equivalent to $ (X_1,Y_1,Z_1)$ via a Clifford $V$, the optimal decomposition is given by
\begin{align}
e^{i\theta_1 P} e^{i\theta_2 Q} e^{i\theta_3 W} = \sum_{i} \alpha_i V^\dagger(C_i\otimes \id\tn{n-1})V.
\end{align}
As before, each Clifford term $V^\dagger(C_i\otimes \id\tn{n-1})V$ is easy to cast as a product of a small number of Cliffords of the form $e^{i\phi R}$ where $R$ is a Pauli and $\phi \in (\pi/4)\bb{Z}$.
This way we obtain a decomposition of the many-body operator $e^{i\theta_1 P} e^{i\theta_2 Q} e^{i\theta_3 W}$ which is both optimal and in a form enabling fast phase-sensitive updates.

The reduction in the 1-norm by decomposing products of gates can be quite large, see for example \Cref{fig:product decompositions}, where the log of the 1-norm is plotted for different one and two qubit gate sequences.

\section{Amplitude estimation}
To illustrate the value of Clifford recompilation, we consider the problem of estimating circuit amplitudes of near-Clifford circuits.
Randomized sparsification can be used to estimate the amplitudes $\inner{x}{\psi}$ of the output state.
Indeed \cref{eq: avg of alg 1} implies that
\begin{align}
\bb{E}(\inner{x}{z}) =\inner{x}{\psi},
\end{align}
where $\ket{z}$ is the random vector generated by \Cref{alg: improved randomized sparsification}.
By a complex-valued version of Hoeffding's inequality  \cite{hoeffding1963probability,Stahlke}, the random vector $\ket{\tilde{\psi}}$ obtained by averaging $k$ instances of $\ket{z}$ satisfies
\begin{align}
\Pr \left[ |\inner{x}{\psi} - \inner{x}{\tilde{\psi}}| \geq \epsilon  \right] \leq \delta,
\end{align}
provided we take $k = O( \norm{a}_1^2 \epsilon^{-2} \log(\delta^{-1}))$.
Thus by averaging enough evaluations of $\inner{x}{z}$, one obtains a good estimate of $\inner{x}{\psi}$ with high probability.
Given the CH-form of $\ket{z}$, it takes time $O(n^2)$ to compute the amplitude $\inner{x}{z}$ using the methods in \cite{bravyi2018simulation}.
This gives a runtime scaling of $O( \norm{a}_1^2 m n^2 \epsilon^{-2} \log(\delta^{-1}))$ for computing an additive error estimate of the amplitude, using the improved sparsification Monte Carlo of \Cref{sec: improved sparsification}.
Typically we are interested in estimating the probability $\Pr(x) = |\inner{x}{\psi}|^2$, rather than the amplitude.
It is straightforward to check that if $|\inner{x}{\tilde{\psi}} - \inner{x}{\psi}| \leq c \epsilon$, where $c=\sqrt{2}-1$, then $\left| |\inner{x}{\tilde{\psi}}|^2  - P(x) \right| \leq \epsilon$ (as long as $\epsilon \leq 1$).
In other words, estimating the probability instead of the amplitude only requires multiplying $k$ by the constant $c^{-2} \approx 6$.

Note that the heuristic Metropolis simulator of \cite{bravyi2018simulation} relies on taking the ratio between two amplitude estimates, obtained exactly as described above, in order to determine the next move at each Metropolis step.

Estimating amplitudes as above relies on very few properties of stabilizer states.
It is only needed that stabilizer states admit an efficient, phase-sensitive description which can be updated under Clifford gates, and that this description can be used to efficiently compute any desired amplitude of the state.
Thus estimating amplitudes as described in this section, and indeed the heuristic Metropolis simulator of \cite{bravyi2018simulation}, can both be used for other states and gates satisfying similar criteria.
Perhaps the main candidate for this generalization is Fermionic Gaussian states and nearest-neighbor matchgates \cite{valiant2002quantum,Terhal2002}.
Nearest-neighbor matchgates acting on Fermionic Gaussian states are known to be classically simulable, and yet augmenting the gate set with nearest-neighbor SWAP gates is sufficient for universal quantum computing.
This motivates searching for matchgate-decompositions of the SWAP gate. 
Our numerical search suggests that the decomposition
\begin{align}\label{eq: matchgate decomposition}
\text{SWAP} = \frac{e^{i \pi /4}}{\sqrt{2}} G( \id, -i X) + \frac{e^{-i \pi 
/4}}{\sqrt{2}} G(\id, iX)
\end{align}
has minimal $1$-norm over all matchgate-decompositions of $\text{SWAP}$, where $G(\id, \pm iX)$ is the $\pm i$SWAP gate (which is a matchgate).
This decomposition has 1-norm equal to $\sqrt{2}$, which implies that amplitude estimation using this decomposition has a runtime scaling like $O(2^m)$, where $m$ is the number of SWAP gates.
Since it is possible to exactly compute the amplitude in time $O(2^m)$ by computing the $2^m$ matchgate contributions corresponding to the expansion in \cref{eq: matchgate decomposition}, this decomposition is not useful for this purpose.
It may be possible nevertheless to find smaller 1-norms by considering decompositions of parallel SWAP gates into nearest-neighbor matchgate circuits involving more than two qubits.   

Applying the same idea to local vs. entangling gates, we find numeric evidence that the decomposition
\begin{align}
\text{CPHASE} = \frac{e^{i\pi4}}{\sqrt{2}} (S_1^\dagger S_2^\dagger - i S_1 S_2)
\end{align}
achieves the minimum 1-norm over all decompositions of CPHASE into product gates (not necessarily Cliffords), where $S_i$ is the PHASE gate on qubit $i$.
Similar to the matchgate case, the 1-norm of this decomposition is too large to provide an improvement over brute force methods.
Nevertheless, our numerics is not exhaustive and a better decomposition may exist. 
It may also be the case that a decomposition of the form
\begin{align}
\text{CPHASE}\tn{k} = \sum_{j} a_j V_j,
\end{align}
exists such that each $V_j$ is a product gate on $2k$ qubits and $\norm{a}_1 < 2^{k/2}$.
Our numerical search did not detect such decompositions for $k=2$ allowing for up to 4 terms in the sum.

\section{Conclusion} We have described how to improve the simulators of Bravyi et al \cite{bravyi2018simulation} by recompiling the circuit using exponentiated Pauli operators. 
The theoretical runtime reduction is given by the factor $\ell / m$ where $\ell$ and $m$ are respectively the total and non-Clifford gate counts, which can be significant when the target circuit is vastly dominated by Clifford gates.
As an added benefit, the recompilation puts the circuit in a form where it is easier to search for optimal Clifford decompositions, since all non-Clifford gates appear in sequence.
This can significantly reduce the exponential prefactor in the runtime.
The circuit recompilation is well-suited to be used with a scheme for computing additive-error estimates of outcome probabilities.
This scheme requires much less work than the full techniques of \cite{bravyi2018simulation}, and can in principle be used with a broader class of states and gates.
 
\hd{Acknowledgements.} 
This research was supported by the U.S. Army Research Office through grant W911NF-14-1-0103,
and by funding from the Government of Ontario through TQT, and the Government of Canada through CFREF and NSERC.
\bibliography{library}

\begin{thebibliography}{21}
\providecommand{\natexlab}[1]{#1}
\providecommand{\url}[1]{\texttt{#1}}
\expandafter\ifx\csname urlstyle\endcsname\relax
  \providecommand{\doi}[1]{doi: #1}\else
  \providecommand{\doi}{doi: \begingroup \urlstyle{rm}\Url}\fi

\bibitem[Bennink et~al.(2017)Bennink, Ferragut, Humble, Laska, Nutaro,
  Pleszkoch, and Pooser]{bennink2017monte}
Ryan~S Bennink, Erik~M Ferragut, Travis~S Humble, Jason~A Laska, James~J
  Nutaro, Mark~G Pleszkoch, and Raphael~C Pooser.
\newblock Monte carlo simulation of near-clifford quantum circuits.
\newblock \emph{arXiv preprint}, 2017.
\newblock URL \url{https://arxiv.org/abs/1703.00111}.

\bibitem[Boixo et~al.(2017)Boixo, Isakov, Smelyanskiy, and
  Neven]{boixo2017simulation}
Sergio Boixo, Sergei~V Isakov, Vadim~N Smelyanskiy, and Hartmut Neven.
\newblock Simulation of low-depth quantum circuits as complex undirected
  graphical models.
\newblock \emph{arXiv preprint arXiv:1712.05384}, 2017.
\newblock URL \url{https://arxiv.org/abs/1712.05384}.

\bibitem[Bravyi and Gosset(2016)]{bravyigosset}
Sergey Bravyi and David Gosset.
\newblock Improved classical simulation of quantum circuits dominated by
  clifford gates.
\newblock \emph{Phys. Rev. Lett.}, 116:\penalty0 250501, Jun 2016.
\newblock \doi{10.1103/PhysRevLett.116.250501}.

\bibitem[Bravyi et~al.(2016)Bravyi, Smith, and Smolin]{bravyi2016trading}
Sergey Bravyi, Graeme Smith, and John~A Smolin.
\newblock Trading classical and quantum computational resources.
\newblock \emph{Physical Review X}, 6\penalty0 (2):\penalty0 021043, 2016.
\newblock \doi{10.1103/physrevx.6.021043}.

\bibitem[Bravyi et~al.(2018)Bravyi, Browne, Calpin, Campbell, Gosset, and
  Howard]{bravyi2018simulation}
Sergey Bravyi, Dan Browne, Padraic Calpin, Earl Campbell, David Gosset, and
  Mark Howard.
\newblock Simulation of quantum circuits by low-rank stabilizer decompositions.
\newblock \emph{arXiv preprint}, 2018.
\newblock URL \url{https://arxiv.org/abs/1808.00128}.

\bibitem[Chen et~al.(2018{\natexlab{a}})Chen, Zhang, Chen, Huang, Newman, and
  Shi]{chen2018classical}
Jianxin Chen, Fang Zhang, Mingcheng Chen, Cupjin Huang, Michael Newman, and
  Yaoyun Shi.
\newblock Classical simulation of intermediate-size quantum circuits.
\newblock \emph{arXiv preprint}, 2018{\natexlab{a}}.
\newblock URL \url{https://arxiv.org/abs/1805.01450}.

\bibitem[Chen et~al.(2018{\natexlab{b}})Chen, Zhou, Xue, Yang, Guo, and
  Guo]{chen201864}
Zhaoyun Chen, Qi~Zhou, Cheng Xue, Xia Yang, Guangcan Guo, and Guoping Guo.
\newblock 64-qubit quantum circuit simulation.
\newblock \emph{Science Bulletin}, 63\penalty0 (15):\penalty0 964--971, aug
  2018{\natexlab{b}}.
\newblock \doi{10.1016/j.scib.2018.06.007}.

\bibitem[De~Raedt et~al.(2018)De~Raedt, Jin, Willsch, Nocon, Yoshioka, Ito,
  Yuan, and Michielsen]{de2018massively}
Hans De~Raedt, Fengping Jin, Dennis Willsch, Madita Nocon, Naoki Yoshioka,
  Nobuyasu Ito, Shengjun Yuan, and Kristel Michielsen.
\newblock Massively parallel quantum computer simulator, eleven years later.
\newblock \emph{arXiv preprint arXiv:1805.04708}, 2018.
\newblock \doi{10.1016/j.cpc.2018.11.005}.

\bibitem[Fowler et~al.(2012)Fowler, Mariantoni, Martinis, and
  Cleland]{Fowler2012}
Austin~G. Fowler, Matteo Mariantoni, John~M. Martinis, and Andrew~N. Cleland.
\newblock Surface codes: Towards practical large-scale quantum computation.
\newblock \emph{Physical Review A}, 86\penalty0 (3), sep 2012.
\newblock \doi{10.1103/physreva.86.032324}.

\bibitem[Garc{\'\i}a et~al.(2014)Garc{\'\i}a, Markov, and
  Cross]{garcia2014geometry}
H{\'e}ctor~J Garc{\'\i}a, Igor~L Markov, and Andrew~W Cross.
\newblock On the geometry of stabilizer states.
\newblock \emph{Quantum Information \& Computation}, 14\penalty0
  (7\&8):\penalty0 683--720, 2014.
\newblock URL \url{https://arxiv.org/abs/1711.07848}.

\bibitem[Harrow and Montanaro(2017)]{Harrow_2017}
Aram~W. Harrow and Ashley Montanaro.
\newblock Quantum computational supremacy.
\newblock \emph{Nature}, 549\penalty0 (7671):\penalty0 203--209, sep 2017.
\newblock \doi{10.1038/nature23458}.

\bibitem[Heinrich and Gross(2018)]{heinrich2018robustness}
Markus Heinrich and David Gross.
\newblock Robustness of magic and symmetries of the stabiliser polytope.
\newblock \emph{arXiv preprint arXiv:1807.10296}, 2018.
\newblock URL \url{https://arxiv.org/abs/1807.10296}.

\bibitem[Hoeffding(1963)]{hoeffding1963probability}
Wassily Hoeffding.
\newblock Probability inequalities for sums of bounded random variables.
\newblock \emph{Journal of the American statistical association}, 58\penalty0
  (301):\penalty0 13--30, 1963.
\newblock \doi{10.1007/978-1-4612-0865-5_26}.

\bibitem[Howard and Campbell(2017)]{howard2017application}
Mark Howard and Earl Campbell.
\newblock Application of a resource theory for magic states to fault-tolerant
  quantum computing.
\newblock \emph{Physical Review Letters}, 118\penalty0 (9):\penalty0 090501,
  2017.
\newblock \doi{10.1103/physrevlett.118.090501}.

\bibitem[Litinski(2018)]{litinski2018game}
Daniel Litinski.
\newblock A game of surface codes: Large-scale quantum computing with lattice
  surgery.
\newblock \emph{arXiv preprint arXiv:1808.02892}, 2018.
\newblock URL \url{https://arxiv.org/abs/1808.02892}.

\bibitem[Neill et~al.(2018)Neill, Roushan, Kechedzhi, Boixo, Isakov,
  Smelyanskiy, Megrant, Chiaro, Dunsworth, Arya, Barends, Burkett, Chen, Chen,
  Fowler, Foxen, Giustina, Graff, Jeffrey, Huang, Kelly, Klimov, Lucero, Mutus,
  Neeley, Quintana, Sank, Vainsencher, Wenner, White, Neven, and
  Martinis]{Neill2018}
C.~Neill, P.~Roushan, K.~Kechedzhi, S.~Boixo, S.~V. Isakov, V.~Smelyanskiy,
  A.~Megrant, B.~Chiaro, A.~Dunsworth, K.~Arya, R.~Barends, B.~Burkett,
  Y.~Chen, Z.~Chen, A.~Fowler, B.~Foxen, M.~Giustina, R.~Graff, E.~Jeffrey,
  T.~Huang, J.~Kelly, P.~Klimov, E.~Lucero, J.~Mutus, M.~Neeley, C.~Quintana,
  D.~Sank, A.~Vainsencher, J.~Wenner, T.~C. White, H.~Neven, and J.~M.
  Martinis.
\newblock A blueprint for demonstrating quantum supremacy with superconducting
  qubits.
\newblock \emph{Science}, 360\penalty0 (6385):\penalty0 195--199, apr 2018.
\newblock \doi{10.1126/science.aao4309}.

\bibitem[Preskill(2018)]{Preskill2018}
John Preskill.
\newblock Quantum computing in the {NISQ} era and beyond.
\newblock \emph{Quantum}, 2:\penalty0 79, aug 2018.
\newblock \doi{10.22331/q-2018-08-06-79}.

\bibitem[Smelyanskiy et~al.(2016)Smelyanskiy, Sawaya, and
  Aspuru-Guzik]{smelyanskiy2016qhipster}
Mikhail Smelyanskiy, Nicolas~PD Sawaya, and Al{\'a}n Aspuru-Guzik.
\newblock qhipster: the quantum high performance software testing environment.
\newblock \emph{arXiv preprint arXiv:1601.07195}, 2016.
\newblock URL \url{https://arxiv.org/abs/1601.07195}.

\bibitem[Stahlke(2014)]{Stahlke}
Dan Stahlke.
\newblock Quantum interference as a resource for quantum speedup.
\newblock \emph{Phys. Rev. A}, 90:\penalty0 022302, Aug 2014.
\newblock \doi{10.1103/PhysRevA.90.022302}.

\bibitem[Terhal and DiVincenzo(2002)]{Terhal2002}
Barbara~M. Terhal and David~P. DiVincenzo.
\newblock Classical simulation of noninteracting-fermion quantum circuits.
\newblock \emph{Physical Review A}, 65\penalty0 (3), mar 2002.
\newblock \doi{10.1103/physreva.65.032325}.

\bibitem[Valiant(2002)]{valiant2002quantum}
Leslie~G Valiant.
\newblock Quantum circuits that can be simulated classically in polynomial
  time.
\newblock \emph{SIAM Journal on Computing}, 31\penalty0 (4):\penalty0
  1229--1254, 2002.
\newblock \doi{10.1137/s0097539700377025}.

\end{thebibliography}
\bibliographystyle{plainnat}

\end{document}